\documentclass[journal]{elsarticle}
\usepackage{amsmath,amsfonts,amssymb,amsthm,enumerate,bm,multirow,verbatim,url,color,comment,color,setspace,bbm,mathdots,mathrsfs,latexsym,times,amsfonts,stmaryrd,caption}
\usepackage{xcolor}
\usepackage{graphicx}
\usepackage[]{times}
\usepackage{epsfig,subfigure,wrapfig}
\usepackage[ruled,vlined]{algorithm2e}
\usepackage{algorithmic}
\newtheorem{theorem}{Theorem}
\newtheorem{definition}{Definition}
\newtheorem{lemma}{Lemma}
\newtheorem{corollary}{Corollary}

\newtheorem{proposition}{Proposition}

\newcommand{\C}{{\mathcal{C}}}
\newcommand{\G}{{G}}
\newcommand{\V}{V}
\newcommand{\e}{E}

\newcommand{\n}{{\mathcal{N}}}
\newcommand{\RR}{{\mathbb{R}}}
\newcommand{\tp}{{T_{\alpha}}}
\newcommand{\tm}{{T_{\beta}}}
\newcommand{\dist}{{\rm dist}}
\newcommand{\prez}{{\mathbf{R_0}}}
\newcommand{\prezz}{{\mathbf{R'_0}}}
\newcommand{\prel}{{\mathbf{R}_{l}}}

\newcommand{\pretty}{{\mathbf{R}}}
\newcommand{\st}{\colon\,}

\newcommand{\sizeof}[1]{\lvert #1 \rvert}
\renewcommand{\subset}{\subseteq}

\begin{document}

\title{Paired Threshold Graphs}

\author{Vida Ravanmehr, Gregory J. Puleo,  Sadegh Bolouki, and Olgica Milenkovi\'c}
\address{Coordinated Science Lab, University of Illinois, Urbana-Champaign}


\cortext[mycorrespondingauthor]{Corresponding author: Olgica Milenkovi\'c}
\ead{milenkov@illinois.edu}


\begin{abstract}
Threshold graphs are recursive deterministic network models that have been proposed for describing certain economic and social interactions. One drawback of this graph family is that it has limited generative attachment rules. To mitigate this problem, we introduce a new class of graphs termed Paired Threshold (PT) graphs described through vertex weights that govern the existence of edges via two inequalities. One inequality imposes the constraint that the sum of weights of adjacent vertices has to exceed a specified threshold. The second inequality ensures that adjacent vertices have a weight difference upper bounded by another threshold. We provide a conceptually simple characterization and decomposition of PT graphs, analyze their forbidden induced subgraphs and  present a method for performing vertex weight assignments on PT graphs that satisfy the defining constraints. Furthermore, we describe a polynomial-time algorithm for recognizing PT graphs. We conclude our exposition with an analysis of the intersection number, diameter and clustering coefficient of PT graphs.
\end{abstract}

\maketitle

\section{Introduction}
The problem of analyzing complex behaviors of large social, economic and biological networks based on generative recursive and probabilistic models has been the subject of intense research in graph theory, machine learning and statistics. In these settings, one often assumes the existence of attachment and preference rules for network formation, or imposes constraints on subgraph structures as well as vertex and edge features that govern the creation of network communities~\cite{jackson2008social,Richards, LimitThreshold, GeogThreshold, RandomThreshold}. Models of this type have been used to predict network dynamics and topology fluctuations, infer network community properties and preferences, determine the bottlenecks and rates of spread of information and commodities and elucidate functional and structural properties of individual network modules~\cite{babis,NIPS2009_3846,ICML2012Palla_785}.  

Here, we propose a new deterministic family of graph structures that may be used for social and economic interaction modeling and easily extended to a probabilistic setting. The graphs in question, termed Paired Threshold (PT) graphs, are succinctly characterized as follows: each vertex is assigned a nonnegative weight. An edge between two vertices exists if and only if the sum of the vertex weights exceeds a certain threshold, and at the same time, the absolute value of the difference between the weights remains bounded by another prescribed threshold. PT graphs are generalizations of two classes of graphs: threshold and unit interval graphs. Threshold graphs were introduced by Chv\'atal and Hammer~\cite{ThresholdGraph} in order to solve a set-packing problem; they are defined by the first generative property of PT graphs, stating that an edge between two vertices exists if and only if the sum of their weights exceeds a predetermined threshold. Threshold graphs are used for aggregation of inequalities, synchronization and cyclic scheduling~\cite{thresholdBook}, as well as for social network modeling~\cite{LimitThreshold,masuda2005geographical}.
The concept of unit interval graphs was first introduced in~\cite{Scott}, based on a characterization of semi-orders (unit interval orders). Unit interval graphs were further investigated by Wegner in his seminal work~\cite{UnitIntervalGraphs}; there, the graphs were described in terms of vertex weights constrained that ensure that the difference of the weights of every pair of adjacent vertices lies below a predefined threshold. Other classes of graphs related to PT graphs include {\cal{quasi threshold graphs}}, introduced in~\cite{QuasiThreshold}; and {\cal{mock threshold}} graphs~\cite{MockThreshold}.


Probabilistic extensions of the deterministic model are possible as well, for example by assuming that the vertices satisfying the two weight constraints are adjacent with high probability, while vertices not satisfying the constraints are adjacent with small probability. Another approach to creating probabilistic PT graphs is to allow the vertex weights to be random variables with some prescribed distribution (e.g., uniform or Gaussian). Random PT graphs will be discussed elsewhere.

The main contributions of this work are proofs establishing a number of properties of PT graphs. First, we show that PT graphs exhibit a special hierarchical distance decomposition involving unit interval graphs and cliques. Second, we exhibit polynomial-time algorithms 
for deciding if a graph is PT or not. Third, we prove that PT graphs have small diameter, avoid ``anti-motifs'' of real social and biological networks as induced subgraphs and include graphs with good clustering coefficients.

The paper is organized as follows. In Section~\ref{Preli}, we briefly review relevant definitions and concepts from graph theory and introduce PT graphs. In Section~\ref{DTG}, we characterize the topological properties of PT graphs and some of their forbidden induced subgraphs, and describe a decomposition of the graphs. This decomposition allows one to find a vertex weight assignment that satisfies the PT graph constraints. In Section~\ref{Algorithm}, we provide a polynomial-time algorithm for identifying whether a graph is PT or not. In Section~\ref{sec:social}, using the previously devised PT graphs decomposition, we first describe a number of forbidden induced subgraphs of PT graphs and then provide  closed formulas for the intersection number and the clustering coefficient of PT graphs as well as a bound on the diameter of PT graphs.
\section{Preliminaries and background}\label{Preli}

We start by introducing relevant definitions and by providing an overview of basic properties of threshold graphs.
\begin{figure}[t]
        \begin{center}
        \subfigure[]{
         \centering\includegraphics[width=0.7in]{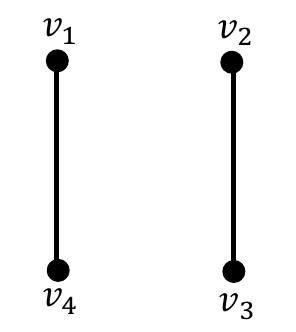}
         }
        \subfigure[]{
               \centering \includegraphics[width=0.7in]{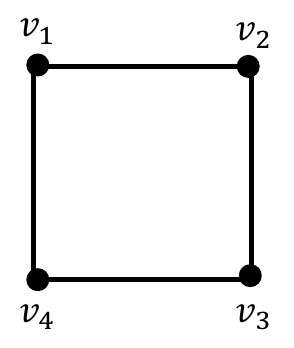}
              
   }
      \subfigure[]{
         \centering\includegraphics[width=1.2in]{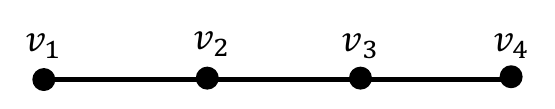}
         }
          
        \end{center}
        \caption{Forbidden induced subgraphs in threshold graphs: (a) $2K_2$ (two disjoint edges), (b) $C_4$ (a cycle of length four), (c) $P_4$ (a path of length four).}
        \vspace{-0.18in}
				\label{fig:tgforbidden}
\end{figure}
Throughout the paper, $\RR$ is used to denote the set of real numbers,
while $\RR^{+}$ is used to denote the set of positive real numbers.

	Let $\G(\V,\e)$ be an undirected graph, with vertex set $\V = \{1,\ldots,n\}$ and edge set $\e$. Two vertices $i,j \in \V$, $i \neq j$, are said to be {\it adjacent} if there exists an edge in $\e$, herein denoted by $e_{ij}$, connecting them. For every $i \in \V$, we denote by $\n(i)$ the set of the vertices adjacent to $i$, i.e.,
	\begin{equation}
		\n(i) \triangleq \{j \in \V \,|\, e_{ij} \in \e \}.
	\end{equation}
The cardinality of $\n(i)$, denoted by $d(i)$, is referred to as the degree of vertex $i$. 
	\begin{definition}
		A graph $\G(\V,\e)$ is called a {\it threshold graph} if there exists a fixed $T \in \RR^{+}$, and a weight function $w:\V \rightarrow \RR^{+},$ such that for all distinct $i,j \in \V$:
		\begin{equation}
			e_{ij} \in \e ~\Leftrightarrow~ w(i) + w(j) \geq T.
		\label{T condition}
		\end{equation}
We refer to such a threshold graph as a $(T,w)$ graph~\cite{thresholdBook}.
	\end{definition}
Threshold graphs may be equivalently defined as those graphs that avoid $C_4$, $P_4$ and $2K_2$ as induced subgraphs \cite{thresholdBook}
(see Figure~\ref{fig:tgforbidden}). Furthermore, threshold graphs may be generated using a recursive procedure, 
by sequentially adding an isolated vertex (a vertex not connected to any previously added vertices) or a dominating vertex (a vertex connected to all previously added vertices)~\cite{thresholdBook}.

	Threshold graphs may also be alternatively characterized via what is called the~\emph{vicinal preorder} $\pretty$~\cite{thresholdBook}, defined on the vertices of $\G$ as:
	\begin{equation}
	 	i \,\pretty\, j ~\Leftrightarrow~ \n(i) \backslash \{j\} \subset \n(j).
	\label{vicinal preorder}
	\end{equation}
The preorder $\pretty$ described in (\ref{vicinal preorder}) is total if it is a binary relation which is transitive and for any pair of vertices $i,j$, one has 
$i \pretty j$ or $j \pretty i$. Given a threshold graph with threshold $T$ and vertex weights $w$, it is straightforward to show that
	\begin{equation}
		i \,\pretty\, j ~\Leftrightarrow~ w(i) \leq w(j).
	\end{equation}
Therefore, since the preorder $\leq$ on the set $\RR^{+}$ is total, the preorder $\pretty$ onthe vertices of $\G$ is total as well. It turns out that the converse is also true \cite{thresholdBook}, i.e., if the preorder $\pretty$ is total, then $\G$ is a threshold graph. To see why this is true, let $\delta_1 < \ldots < \delta_m$ represent all the distinct, positive degrees of the vertices of $\G$, and set $\delta_0 = 0$. For all $i$, $0 \leq i \leq m$, define
	\begin{equation}
		D_i \triangleq \{i \in \V \,|\, d(i) = \delta_i\}.
		\label{DegreeDist}
	\end{equation}
Notice that $(D_0,\ldots,D_m)$ forms a partition\footnote{With a slight abuse of terminology, we use the term ``partition'' although $D_0$ may be empty.} of $\V$, known as the {\it degree partition} of $\V$. Define the vertex weight function $w$ according to $w(i) = j$, $\forall i \in D_j$, $0 \leq j \leq m$, and set the threshold to $T = m+1$. One can then show that the threshold $T$ and the aforedescribed weight function $w$ satisfy (\ref{T condition}), implying that $\G$ is a threshold graph \cite{thresholdBook}.

\begin{proposition}\label{good old}
	A graph $\G(\V,\e)$ is a threshold graph if and only if the preorder $\pretty$ defined in (\ref{vicinal preorder}) is total.
	\end{proposition}
	
Unit interval graphs are defined as follows.
	\begin{definition}
		A graph $\G(\V,\e)$ is called a {\it unit interval graph} if there exist a fixed $T \in \RR^{+}$, and a weight function $w:\V \rightarrow \RR^{+}$ such that for all distinct $i,j \in \V,$
		\begin{equation}
			e_{ij} \in \e ~\Leftrightarrow~ |w(i) - w(j)| \leq T.
		\label{I condition}
		\end{equation}
	\end{definition}
\begin{definition}
		Given a connected graph $\G(\V,\e),$ a {\it distance decomposition} of $\V$ is a partition $(\C_0,\C_1,\ldots,\C_m)$, $m \geq 0$, of $\V$ in which
		\begin{equation} \label{eq:cliques}
			\C_l \triangleq \left\{ i \in \V \,{\Big{|}}\,  \dist(i,j)_{j \in \C_{0}} = l\right\},~\forall l,~ 1 \leq l \leq m,
		\end{equation}
		where $\dist(i,j)$ is the length of the shortest path between $i$ and $j$ in the graph $G$.
	\end{definition}
	Equivalently, a distance decomposition may be generated starting from a set $\C_0$, and then recursively creating $\C_l$, $1 \leq l \leq m$, according to 
		\begin{equation}
			\C_l \triangleq \left\{ i \in \V\backslash \bigcup_{l'=0}^{l-1}\C_{l'} \,{\Big{|}}\, \exists j \in \C_{l-1}:~e_{ij}\in \e \right\}.
		\label{C rec}
		\end{equation}
Simply put, $\C_1$ is the set of vertices adjacent to $\C_0$ in $\G$, excluding $\C_0$; $\C_2$ is the set of vertices adjacent to $\C_1$ in $\G$, excluding $\C_0$ and $\C_1$, and so on. Clearly, there is no edge between $\C_l$ and $\C_{l'}$, $0 \leq l,l' \leq m$, if $|l-l'| \geq 2$. 
	
	We introduce next a new family of graphs, termed {\it paired threshold} graphs, which combine the properties of threshold and unit interval graphs.	
	\begin{definition}
A graph $\G(\V,\e)$ is termed a paired threshold (PT) graph if there exist two fixed thresholds $\tp\geq \tm \in \RR^{+}$ and a weight function $w:\V \rightarrow \RR^{+}$, such that for all distinct $i,j \in \V$,
		\begin{equation}
		e_{ij} \in \e ~\Leftrightarrow~
			\begin{cases}
				w(i) + w(j) \geq \tp,\\
				and\\
				|w(i)-w(j)| \leq \tm.
			\end{cases}
		\label{DT condition}
		\end{equation}
We will refer to graphs with the above defining properties as $(\tp,\tm,w)$-PT graphs.
	\end{definition}
		Figure \ref{DT} illustrates a PT graph with $\tp=10$ and $\tm=2$, along with a possible weight assignment. Note that there exists an
		edge between the two vertices labeled by $5$ and $7$, as $5+7=12>\tp=10$ and $|5-7|=2 \leq \tm=2$, but there is no edge between the
		vertices labeled by $4$ and $7$ as $|4-7|=3>\tm=2$.
	\begin{figure}[t]
		\centering
		\includegraphics[width=2.3in]{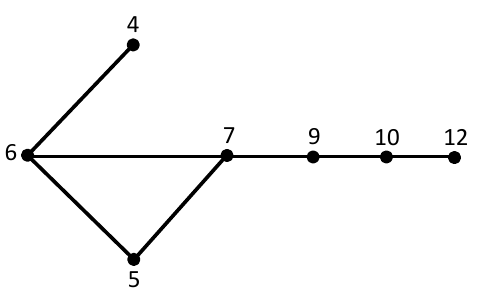}
		\caption{An example of a PT graph, along with a weight assignment for the parameters $\tp=10$ and $\tm=2$.}
	\label{DT}
	\vspace{-0.12in}
	\end{figure}
\section{Characterization of PT Graphs}\label{DTG}

We characterize next the structure of a general connected PT graph $\G(\V,\e)$ with $|\V| \geq 2$ and parameters $(\tp,\tm,w)$. The main result is stated in Theorem~\ref{maint}.

\begin{theorem}\label{maint}
          A connected graph $\G(\V,\e)$ is a PT graph if and only if
          it is a unit interval graph or if there is a distance
          decomposition $(\C_0,\C_1,\ldots,\C_m)$, for some
          $m \geq 0$, for which all the following statements hold
          true:
			\begin{enumerate}[(i)]
				\item The vicinal preorder $\prez$ defined on the elements of $\C_0$ as
				\begin{equation}
						i \,\prez\, j ~\Leftrightarrow~ \n(i) \backslash \{j\} \subseteq \n(j),
						\label{eq1}
					\end{equation}
is total.\label{con1}
				\item For every $l$, $1 \leq l \leq m$, the subgraph of $\G$ induced by $\C_l$ is a clique.\label{con2}
				\item The  preorder $\prel$ defined on the elements of $\C_l$, $1 \leq l \leq m$ according to
					\begin{equation}
						i \,\prel\, j ~\Leftrightarrow~
						\begin{cases}
							\n(j)\cap \C_{l-1}  \subseteq \n(i) \cap \C_{l-1},\\
							and\\
							\n(i)\cap \C_{l+1}  \subseteq \n(j) \cap \C_{l+1},
						\end{cases}
					\label{hoof}
					\end{equation}
is total; here, we enforce $\C_{m+1}=\emptyset$.\label{con3}
			\end{enumerate}
	\end{theorem}

We start by proving the ``only if'' part of the theorem through a series of intermediate results described in Propositions~\ref{no loss 1}-\ref{fact 4}.	

First, note that for every $e_{ij} \in \e$, from the two inequalities in (\ref{DT condition}), one must have $\min \{w(i),w(j)\} \geq \frac{\tp-\tm}{2}$. Thus, noticing that every vertex has at least one neighbor as the graph is connected, we have the following proposition. 
\begin{proposition}\label{no loss 1}
		If $\G$ is a connected graph with at least two vertices, then for every $i \in \V$, $w(i) \geq \frac{\tp-\tm}{2}$.
	\end{proposition}
	
	We now proceed to demonstrate that if $\G$ is not a unit interval graph, its set of vertices $\V$ has a distance decomposition $(\C_0,\ldots,\C_m)$ with a special structure. For this purpose, we define
	\begin{equation}
		\C_0 \triangleq \left\{ i\in \V\,|\,w(i) \in \left[\frac{\tp-\tm}{2} , \frac{\tp+\tm}{2}\right) \right\}.
	\label{C zero}
	\end{equation}

	\begin{proposition}\label{fact 1}
		The subgraph induced by $\V\backslash\C_0$ is a unit interval graph with parameters $(\tm,w)$. Consequently, if $\C_0$ is the empty set, then $\G$ is a unit interval graph.
	\end{proposition}
	\begin{proof}

As $|w(i)-w(j)| \leq \tm$ holds for every edge in the subgraph induced by $\V\backslash\C_0$, it suffices to show that $w(i)+w(j) \geq \tp$ for all  $i, j \in \V\backslash\C_0$. But this inequality follows by simply noting that  according to Proposition~\ref{no loss 1} and the definition of $\C_0$, $w(i)$ and $w(j)$ are both greater than or equal to $\frac{\tp+\tm}{2}$.
	\end{proof}
	\vspace{-0.05in}
	\begin{proposition}
	  Suppose that $\C_0$ is non-empty. Then, for any $i,j \in \C_0$, one can assume that
	\begin{equation}
		\n(i) \backslash \{j\} = \n(j) \backslash \{i\} ~\Rightarrow~ w(i) = w(j).
	\label{no loss 2}
	\end{equation}
	\label{newProp}
	\end{proposition}
	\begin{proof}
	 Assume that for some $i,j \in \C_0$, (\ref{no loss 2}) does not hold, i.e., that $\n(i) \backslash \{j\} = \n(j) \backslash \{i\}$ but $w(i) \neq w(j)$. Then, one can modify the weights assigned to $i$ and $j$ so as to satsify $w(i) = w(j)$. The modified weight assignment for $i$ and $j$ equals
		\begin{equation}
		w(i) = w(j)=
			\begin{cases}
				\max\{w(i),w(j)\},~ {\rm{if}}~ e_{ij} \in \e,\\
				\min\{w(i),w(j)\}, ~{\rm{if}} e_{ij} \not\in \e.\\
			\end{cases}
		\label{WeightModif}
		\end{equation}
	 It is straightforward to check that the constraints on the weights of vertices of PT graphs still hold under the modified weight assignment, and hence the graph topology remains unchanged. To see why the weight reassignment approach described above terminates, we first note that during the reassignment, the weight of any vertex $i$ changes monotonically. Assume on the contrary that there exists a vertex $i$ whose weight does not change monotonically. Based on (14), there must exist vertices $j$ and $k$, where $e_{ij} \in E$ and $e_{ik} \not\in E$, such that
\begin{equation}
	\n(i)\backslash \{j\} = \n(j) \backslash \{i\}
\label{last1}
\end{equation}
and
\begin{equation}
\n(i) \backslash \{k\} = \n(k) \backslash \{i\}.
\label{last2}
\end{equation}
Now, since $j \in \n(i)$, (\ref{last2}) yields $j \in \n(k)$, and consequently, $k \in \n(j)$. This, together with (\ref{last1}), results in $k \in \n(i)$ and a contradiction, since $e_{ik} \not\in E$. Given (\ref{WeightModif}) and the fact that the weights change monotonically over the course of the weight reassignment process, it follows that the weights can only take finitely many values. Thus, the weight reassignment process terminates in finite time.
	\end{proof}	
	Modifying the weights as described in~(\ref{WeightModif}) for all $i,j \in \C_0$ for which $\n(i) \backslash \{j\} = \n(j) \backslash \{i\}$ but $w(i) \neq w(j)$ results in a weight assignment $w$ for which~(\ref{no loss 2}) is satisfied for every $i,j \in \C_0$.
	Furthermore, it may be assumed without loss of generality that~(\ref{no loss 2}) holds for every $i,j \in \V$ for which $e_{ij} \in \e$. In fact, if for some $i,j\in\V, e_{ij}\in\e$,  (\ref{no loss 2}) is violated, one may change the weights assigned to $i$ and $j$ to $\max\{w(i),w(j)\}$ and repeat the reassignment procedure until (\ref{no loss 2}) is satisfied for all $i,j\in\V$ where $e_{ij}\in\e$.
		
	Having defined $\C_0$ in (\ref{C zero}), let $(\C_0,\ldots,\C_m)$ be the distance decomposition of $\V$ starting with $\C_0$ as previously defined. Then, the following result holds.	
	\begin{proposition}\label{fact 2}
		The vicinal preorder $\prez$ defined on the elements of $\C_0$ as
		\begin{equation}
			i \,\prez\, j ~\Leftrightarrow~ \n(i) \backslash \{j\} \subseteq \n(j),
		\label{pre zero}
		\end{equation}
is total.
	\end{proposition}
	
	\begin{proof}
		For the preorder $\prez$ to be total, it suffices to show that for every distinct $i,j \in \C_0$, one has to have
			\begin{equation}
				i \,\prez\, j ~\Leftrightarrow~ w(i) \leq w(j).
			\label{equi total}
			\end{equation}
(We recall that the preorder $\leq$ is total on $\RR^{+}$). From (\ref{pre zero}) and (\ref{equi total}), 
it therefore suffices to prove that for every distinct pair $i,j \in \C_0,$ one has
			\begin{equation}
				w(i) \leq w(j) ~\Leftrightarrow~ \n(i) \backslash \{j\} \subseteq \n(j).
			\end{equation}
		
		\noindent ($\Rightarrow$): Assume that $w(i) \leq w(j)$. We prove for every $k \in \V\backslash \{j\}$ the following fact: if $e_{ik} \in \e$, then $e_{jk} \in \e$. We consider two different cases.
		\begin{itemize}
			\item[1.] If $k \in \C_0$, from the definition of $\C_0$ in (\ref{C zero}), $|w(j)-w(k)| \leq \tm$. Moreover, since $e_{ik} \in \e$, $w(i) + w(k) \geq \tp$. Thus, $w(j) + w(k) \geq \tp$ and hence $e_{jk} \in \e$.\vspace{.1in}
			\item[2.] If $k \in \V \backslash \C_0$, $w(i) \leq w(j) < w(k)$. Since $e_{ik} \in \e$, we have
				\begin{equation}
					w(j) + w(k) \geq w(i) + w(k) \geq \tp,
				\end{equation}
and
				\begin{equation}
					|w(j) - w(k)| \leq |w(i) - w(k)| \leq \tm,
				\end{equation}
which together imply that $e_{jk} \in \e$.
		\end{itemize}

		\noindent ($\Leftarrow$): Assume $\n(i) \backslash \{j\} \subseteq \n(j)$. We prove that $w(i) \leq w(j)$.  If $\n(i) \backslash \{j\} = \n(j) \backslash \{i\}$, from (\ref{no loss 2}), we have $w(i) = w(j)$. Thus, assume that $\n(i) \backslash \{j\}$ is properly contained in $\n(j)\backslash\{i\}$. Then, there exists $k \in \V \backslash \{i,j\}$ such that $e_{ik} \not\in \e$ and $e_{jk} \in \e$. We show that $w(i) < w(j)$ by considering the following two cases.
		\begin{itemize}
			\item[1.] If $k \in \C_0$, from the definition of $\C_0$ in (\ref{C zero}), both $|w(i) - w(k)| \leq \tm$ and $|w(j) - w(k)| \leq \tm$ are satisfied. 
			Thus, since $e_{ik} \not\in \e$ and $e_{jk} \in \e$, according to~(\ref{DT condition}), we must have $w(i) + w(k) < \tp$ and $w(j) + w(k) \geq \tp$, which immediately results in $w(i) < w(j)$.\vspace{.1in}	
			\item[2.] If $k \in \V \backslash \C_0$, from Proposition \ref{no loss 1}  and the definition of $\C_0$ in (\ref{C zero}), $w(k) \geq \frac{\tp+\tm}{2}$. On the other hand, since $i,j \in \C_0$, both $w(i)$ and $w(j)$ are greater than or equal to $\frac{\tp-\tm}{2}$. Thus,  $w(i) + w(k) \geq \tp$ and $w(j) + w(k) \geq \tp$. Therefore, since $e_{ik} \not\in \e$ and $e_{jk} \in \e$, according to (\ref{DT condition}), we must have $|w(i) - w(k)| > \tm$ and $|w(j) - w(k)| \leq \tm$. Recall that $w(k) \geq \frac{\tp+\tm}{2}$, which implies that $w(k)>\max\{w(i),w(j)\}$. Hence, $w(k)-w(i) > \tm$ and $w(k) - w(j) \leq \tm$, which together imply $w(i) < w(j)$.
		\end{itemize}
	\end{proof}
		Next, we give a characterization of the subgraphs induced by $\C_l$ and define a preorder on the vertices in $\C_l$ for all $1\leq l \leq m$ in Propositions \ref{fact 3} and \ref{fact 4}. The proofs of both Propositions follow directly from properties of unit interval graphs and the fact that the subgraph induced by $\V\backslash\C_0$ is a unit interval graph with parameters $(\tm,w)$.
\begin{proposition}\label{fact 3}
		For every $l$, $1 \leq l \leq m$, the subgraph of $\G$ induced by $\C_l$ is a clique.
	\end{proposition}
	\begin{proof}
		First, recall that $\C_0$ contains all vertices whose weight is less than $(\tp+\tm)/2$. Let $l$, $1 \leq l \leq m$, be arbitrary. From the recursive relation (\ref{C rec}) and from conditions in (\ref{DT condition}), it is easy to see that for every $i \in \C_l$, one must have
			\begin{equation}
				\max_{k \in \C_{l-1}} w(k) < w(i) \leq \max_{k \in \C_{l-1}} w(k) + \tm.
			\label{C rec 2}
			\end{equation}
This immediately implies that $|w(i) - w(j)| < \tm$ for every $i,j \in \C_l$. Furthermore, recalling once again that $\C_0$ contains all vertices of weight less than $(\tp+\tm)/2$, we have
		\begin{equation}
			w(i) \geq \frac{\tp+\tm}{2},~\forall i \in \C_l.
		\end{equation}
Thus, for every $i,j \in \C_l$, it holds that
		\begin{equation}
			w(i) + w(j) \geq \frac{\tp+\tm}{2} + \frac{\tp+\tm}{2} \geq \tp.
    		\end{equation}
Therefore, both conditions of (\ref{DT condition}) are satisfied for every $i,j \in \C_l$, which results in $e_{ij} \in \e$, $\forall i,j \in \C_l$. Hence, the subgraph induced by $\C_l$, $1 \leq l \leq m$, is a clique.
	\end{proof}
	\begin{proposition}\label{fact 4}
		The  preorder $\prel$, defined on the elements of $\C_l$, $1 \leq l \leq m$, according to
		\begin{equation}
			i \,\prel\, j ~\Leftrightarrow~
			\begin{cases}
				\n(j)\cap \C_{l-1}  \subseteq \n(i) \cap \C_{l-1},\\
				and\\
				\n(i)\cap \C_{l+1}  \subseteq \n(j) \cap \C_{l+1},
			\end{cases}
		\end{equation}
is total.
	\end{proposition}
\begin{proof}
		Since the preorder $\leq$ on $\RR^{+}$ is total, it suffices to show that:
		\begin{equation}
			w(i) \leq w(j) ~\Leftrightarrow~
			\begin{cases}
				\n(j)\cap \C_{l-1}  \subseteq \n(i) \cap \C_{l-1},\\
				and\\
				\n(i)\cap \C_{l+1}  \subseteq \n(j) \cap \C_{l+1}.
			\end{cases}
		\label{akharashim}
		\end{equation}
		($\Rightarrow$): Assume that $w(i) \leq w(j)$. We first show that $\n(j)\cap \C_{l-1}  \subseteq \n(i) \cap \C_{l-1}$. Let $k \in \n(j)\cap \C_{l-1}$ be arbitrary. Since $k \in \C_{l-1}$, we must have
		$$w(k) < w(i) \leq w(j).$$ 
		On the other hand, since $k \in \n(j)$, we also have $w(j) - w(k) \leq \tm$. Thus, $w(i) - w(k) \leq \tm$. Moreover, since 
		$$w(i) \geq \frac{\tp + \tm}{2}~  {\rm{ and }} ~w(k) \geq \frac{\tp - \tm}{2},$$
		we have $w(i) + w(k) \geq \tp$. Hence, according to (\ref{DT condition}), $e_{ik} \in \e$. 
		
		We show next that $\n(i)\cap \C_{l+1}  \subseteq \n(j) \cap \C_{l+1}$. For an arbitrary $k \in \n(i)\cap \C_{l+1}$, similar to the previous argument, we 
		have $w(i) \leq w(j) < w(k)$ and $w(k) - w(i) \leq \tm$. Thus, $w(k) - w(j) \leq \tm$. Moreover, $w(k) + w(j) \geq \tp$, and according to (\ref{DT condition}),  $e_{jk} \in \e$.
		
		($\Leftarrow$): Assume that both inclusion relations of (\ref{akharashim}) hold. Moreover, assume to the contrary of the claimed assumption that $w(j) < w(i)$. From part ($\Rightarrow$) of the proof, we conclude
		\begin{equation}
			\begin{cases}
				\n(i)\cap \C_{l-1}  \subseteq \n(j) \cap \C_{l-1},\\
				and\\
				\n(j)\cap \C_{l+1}  \subseteq \n(i) \cap \C_{l+1}.
			\end{cases}
		\label{akharashim2}
		\end{equation}	
From (\ref{akharashim2}) and the two inclusion relations of (\ref{akharashim}), we obtain
		\begin{equation}
			\n(i) \cap (\C_{l-1} \cup \C_{l+1}) = \n(j) \cap (\C_{l-1} \cup \C_{l+1}).
		\label{akharashim3}
		\end{equation}
Next, recall that since $i,j \in \C_l$, their neighbors can only be in $\C_{l-1}$, $\C_l$, and $\C_{l+1}$, where the subgraph induced by $\C_l$ is a clique. Thus, from (\ref{akharashim3}), we conclude that $\n(i) \backslash \{j\} = \n(j) \backslash \{i\}$. According to~(\ref{no loss 2}), we must have $w(i) = w(j)$, and the claim follows by contradiction.
\end{proof}

\begin{corollary}
 Let $\G(\V,\e)$ be a PT graph with parameters $(\tp,\tm,w)$ and  let $v \in \V$. The subgraph induced by $S = \{z \in \n(v): w(z) \geq w(v)\}$ is a clique in $\G$.
\label{RightNeighborsClique}
\end{corollary}
We omit the proof of the corollary, as it is a straightforward consequence of the properties of unit interval graphs and since it can be proved similarly to Proposition~\ref{fact 3}.

	\begin{figure}[t]
		\centering
		\includegraphics[width=5in]{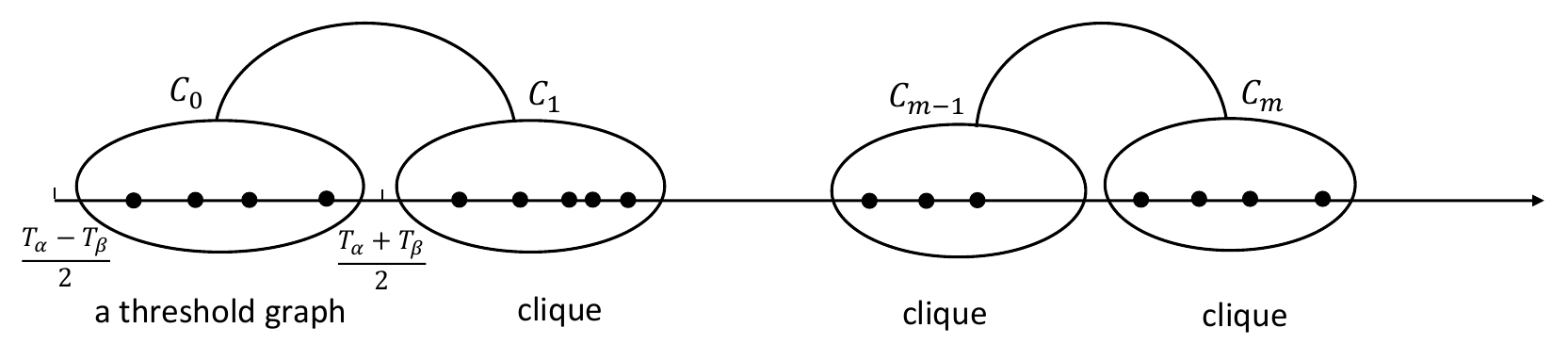}
		\caption{Decompositional structure of a PT graph.}
	\label{partition}
	\end{figure}
	
A distance decomposition of a PT graph is shown in Figure~\ref{partition}. 

In what follows, we prove the ``if'' part of Theorem~\ref{maint} by showing that the PT graph properties established in Propositions \ref{fact 1}-\ref{fact 4} are also sufficient for a graph to be a connected PT graph.
	
	Let $\tp\geq\tm > 0$ be arbitrary. If $\G$ is a unit interval graph, there is a weight function $w:\V \rightarrow \RR^{+}$ such that $\G$ is a unit interval graph with parameters $(\tm,w)$. By defining $w' = w + \frac{\tp}{2}$, it is straightforward to conclude that $\G$ is a $(\tp,\tm,w')$-PT graph. 
	Assume that a distance decomposition $(\C_0,\C_1,\ldots,\C_m)$, where $m \geq 0$, exists and satisfies (\ref{con1})-(\ref{con3}). We construct a weight function $w: \V \rightarrow \RR^{+}$ that establishes that $\G(\V,\e)$ is a $(\tp,\tm,w)$-PT graph. We first assign weights to the vertices in $\C_0$ and then proceed to make similar assignments for the sets $\C_l$, $1 \leq l \leq m$.

		{\it Step 1:} For the weight assignments of $\C_0$, we first show that the subgraph of $\G$ induced by $\C_0$ is a threshold graph. Defining a preorder $\prezz$ on the elements of $\C_0$ according to
		\begin{equation}
			i \,\prezz\, j ~\Leftrightarrow~ \left(\n(i)\cap\C_0\right) \backslash \{j\} \subseteq \n(j) \cap \C_0,
		\end{equation}
we have
		\begin{equation}
			i \,\prez\, j ~\Rightarrow~ i \,\prezz\, j.
		\end{equation}
Thus, since $\prez$ is total according to (\ref{con1}), $\prezz$ is also a total order. Therefore, according to Proposition~\ref{good old}, the subgraph of $\G$ induced by $\C_0$ is a threshold graph.

For the second part of the proof, we need the following lemma.
  \begin{lemma} \label{lem 0}
   For $\tp \geq \tm \in \RR^{+}$, and for all $i \in \C_0$, there exist weight assignments $w(i)$ with the following properties.
		\begin{enumerate}
			\item The subgraph of $\G$ induced by $\C_0$ is a threshold graph with parameters $(\tp,w)$.
			\item For all $i \neq j$, $i,j \in \C_0$, $w(i) \neq w(j)$.
			\item For all $i \in \C_0$,$ w(i) \in \left(\frac{\tp-\tm}{2},\frac{\tp+\tm}{2}\right)$.
		\end{enumerate}
	\end{lemma}
	\begin{proof}
	Recall the notion of the degree partition of the vertices of a graph from the argument leading to Proposition~\ref{good old}. Let $(D_0,\ldots,D_{m'})$ be the degree partition of $\C_0$ in the subgraph of $\G$ induced by $\C_0$. We start with defining the weight function $w:\C_0\rightarrow\RR^+$ as $w(i) = j$ for every $i \in D_j$, $0\leq j \leq m'$.  The subgraph of $\G$ induced by $\C_0$ is a threshold graph with parameters $(m'+1,w)$. We now modify, via the following steps, the weight function $w$ in such a way that it meets the criteria 1-3 of Lemma~\ref{lem 0}.\\
	\it Step 1: For every $i \in \C_0$, we modify $w(i)$ to $w(i)+\epsilon_i$, where $0<\epsilon_i<1/2$, in such a way that the modified weights of every two distinct vertices in $\C_0$ are different. The subgraph of $\G$ induced by $\C_0$ remains a threshold graph with parameters $(m'+1,w)$, a fact which may be verified by observing that $m'+1$ is an integer; the starting weights of the assignment were all integer-valued; and the modified weights are obtained from the previous weights by adding to them a value smaller than $1/2$.\\
	\it Step 2: We next divide  all the weights obtained in the previous step by $m'+1$, to obtain a threshold graph with parameters $(1,w)$, where $w(i) \neq w(j)$ for every distinct $i,j \in \C_0$, and where all the weights are in $(0,1)$.\\
	\it Step 3: Finally, we multiply the weights by $\tm$ and then add $\frac{\tp-\tm}{2}$ to them. It is straightforward to see that the subgraph of $\G$ induced by $\C_0$ becomes a threshold graph with parameters $(\tp,w),$ where $w$ satisfies all the three criteria of Lemma~\ref{lem 0}. 
	\end{proof}
	In conclusion, the weight assignments of $\C_0$ meet all three criteria of Lemma \ref{lem 0}. We also point out that $\forall i,j \in \C_0$,
		\begin{equation}
			w(i) \leq w(j) ~\Rightarrow~ i \prez j.
			\label{hh}
		\end{equation}

		{\it Step 2:} Let a constant $\epsilon > 0$ be such that it satisfies the following two inequalities.
		\begin{equation}
			\epsilon < \min_{i,j \in \C_0}  \Big\{|w(i)-w(j)| \,{\Big{|}}\, w(i) \neq w(j) \Big\},
		\label{dar1}
		\end{equation}
		\begin{equation}
			\epsilon < \frac{n}{n+1}\, \min_{i \in \C_0} \left\{w(i) - \frac{\tp-\tm}{2}\right\}.
		\label{hoof4}
		\end{equation}
Note that since $w$ satisfies Criteria 2 and 3 of Lemma \ref{lem 0}, an $\epsilon>0$ such as described above exists. Then, for every $l$, $1 \leq l \leq m$, we define the vertex weights for $\C_l$ recursively as follows: $\forall i \in \C_l$, set
		\begin{equation}
			\begin{array}{ll}
				w(i) \triangleq
				& \hspace{-.1in}\tm+ \Big( \min \big\{ w(k)\,|\, k\in \n(i) \cap \C_{l-1} \big\} \Big)\vspace{.05in}\\
				& \hspace{-.1in} -\frac{\epsilon}{(n+1)^{l-1}}\left( 1 - \frac{|\n(i) \cap \C_{l+1}|}{n+1} \right),
			\end{array}
		\label{hoof6}
		\end{equation}
and recall that $\C_{m+1}$ is the empty set. Observing that~(\ref{hh}) holds for every $i,j \in \C_0$, by induction on $l$, it is clear from~(\ref{hoof}) that for every $i,j \in \C_l$, $1 \leq l \leq m$,
		\begin{equation}
			w(i) \leq w(j) ~\Leftrightarrow~ i \prel j.
		\label{dar10}
		\end{equation}

	Having defined the vertex weights, we are now ready to prove that $\G(\V,\e)$ is an $(\tp,\tm,w)$-PT graph, i.e., that the Condition~(\ref{DT condition}) is satisfied for every distinct pair of vertices $i,j \in \V$. We consider the following cases.
	
		{\it Case 1:} Let $i,j \in \C_0$. We know that the subgraph of $\G$ induced by $\C_0$ is a threshold graph with parameters $(\tp,w)$. Therefore,
			\begin{equation}
				e_{ij} \in \e \Leftrightarrow w(i) + w(j) \geq \tp.
			\label{hoof2}
			\end{equation}
By noticing from the third criterion of Lemma \ref{lem 0} that both $w(i)$ and $w(j)$ lie in the interval $\left(\frac{\tp-\tm}{2},\frac{\tp+\tm}{2}\right)$, we have $|w(i) - w(j)| \leq \tm$. This fact, together with (\ref{hoof2}), implies (\ref{DT condition}).

		{\it Case 2:} Let $i \in \V\backslash \C_0$. We first state and prove the following lemmas.
		\begin{lemma}\label{lem1}
			For every $\C_l$, $0\leq l \leq m$, and every $k' \in \C_l$, we have
				\begin{equation}
					\frac{\tp+(2l-1)\tm}{2}+ \frac{\epsilon}{n(n+1)^{l-1}}  <w(k') < \frac{\tp+(2l+1)\tm}{2}.
				\label{hoof3}
				\end{equation}
		\end{lemma}
		\begin{proof}
			We prove the inequalities in (\ref{hoof3}) by induction on $l$. For $l=0$, the first inequality of (\ref{hoof3}) is an immediate result of (\ref{hoof4}), while the second inequality follows from the third criterion of Lemma \ref{lem 0}. We now assume that (\ref{hoof3}) holds for $l-1$,  $1 \leq l \leq m$, and prove that it also holds for $l$. To prove the first inequality of (\ref{hoof3}), we observe that
			\begin{equation}
				\begin{array}{l}
					\min \big\{ w(k)\,|\, k\in \n(k') \cap \C_{l-1} \big\}\vspace{.05in}\\
					\geq \min_{k \in \C_{l-1}} w(k)\vspace{.05in}\\
					> \frac{\tp+(2l-3)\tm}{2}+ \frac{\epsilon}{n(n+1)^{l-2}},
				\end{array}
			\label{hoof5}
			\end{equation}
where in the second inequality of (\ref{hoof5}), we used the induction hypothesis. Furthermore,
			\begin{equation}
				|\n(k') \cap \C_{l+1}| \geq 0.
				\label{hoofhoof}
			\end{equation}
Using inequalities (\ref{hoof5}) and (\ref{hoofhoof}) in the recursive relation (\ref{hoof6}) results in the first inequality of (\ref{hoof3}). For the second inequality, by noticing that $|\n(i)\cup \C_{l+1}| \leq n$, one may use (\ref{hoof6}) to obtain 
			\begin{equation}
				\begin{array}{ll}
					w(k')
					&\hspace{-.1in}\leq \tm+ \min_{k\in \C_{l-1}}w(k)\vspace{.05in}\\
					&\hspace{-.1in}< \tm + \frac{\tp+(2l-1)\tm}{2} = \frac{\tp+(2l+1)\tm}{2}.
				\end{array}
			\label{hoof7}
			\end{equation}
In the second inequality, we used the induction hypothesis for $l-1$. 
		\end{proof}
		
	\begin{lemma}\label{lem2}
			For every $\C_l$, $0 \leq l \leq m$, we have
			\begin{equation}
				\frac{\epsilon}{(n+1)^{l}} \leq \min_{k',k'' \in \C_l} \Big\{|w(k')-w(k'')| \,{\Big{|}}\, w(k') \neq w(k'')\Big\}.
			\label{dar2}
			\end{equation}
		\end{lemma}
	\begin{proof}
			The proof follows by induction on $l$. For $l=0$, (\ref{dar2}) reduces to (\ref{dar1}). We now assume that (\ref{dar2}) holds for some $l-1$, $1 \leq l\leq m$ and prove it for $l$. 
			
			First, note that according to~(\ref{hoof6}):
	\begin{equation}
			\begin{array}{ll}\nonumber
			  & w(k')-w(k'')\\
				& = \min \big\{ w(k)\,|\, k\in \n(k') \cap \C_{l-1} \big\} - \min \big\{ w(k)\,|\, k\in \n(k'') \cap \C_{l-1} \big\}\vspace{.05in}\\
				& \hspace{-.1in} + \frac{\epsilon}{(n+1)^{l}}\left( |\n(k') \cap \C_{l+1}|- |\n(k'') \cap \C_{l+1}| \right).
			\end{array}
			\end{equation}
			Thus, in order to have $|w(k')-w(k'')| > 0$,  at least one of the following relations must hold:
			\begin{equation}
				\n(k') \cap \C_{l-1} \neq \n(k'') \cap \C_{l-1},
			\label{dar3}
			\end{equation}
			\begin{equation}
				|\n(k') \cap \C_{l+1}| \neq |\n(k'') \cap \C_{l+1}|
			\label{dar4}
			\end{equation}
Recalling (\ref{con3}), $\prel$ as defined in (\ref{hoof}) is total on $\C_l$. Without loss of generality, assume that $k' \prel k''$, which results in
			\begin{equation}
				\n(k'') \cap \C_{l-1} \subseteq \n(k') \cap \C_{l-1},\\
			   \label{dar5}
			\end{equation}
			\begin{equation}
			  and \nonumber
			\end{equation}
			\begin{equation}
				|\n(k') \cap \C_{l+1}| \leq |\n(k'') \cap \C_{l+1}|.
			\label{dar6}
			\end{equation}

				{\it Case 1:} If (\ref{dar3}) holds, from (\ref{dar5}) one has
					\begin{equation}
						\n(k'') \cap \C_{l-1} \subset \n(k') \cap \C_{l-1},
					\end{equation}
which implies that
					\begin{equation}
						\min \big\{ w(k)\,{\Big{|}}\, k\in \n(k') \cap \C_{l-1} \big\}
						< \min \big\{ w(k)\,{\Big{|}}\, k\in \n(k'') \cap \C_{l-1} \big\}.
					\label{dar7}
					\end{equation}
Notice that the difference between the two expressions on the opposite side of inequality~(\ref{dar7}) is at least $\epsilon/(n+1)^{l-1}$ by the induction hypothesis. Using this observation and~(\ref{dar6}) in the recursive relation~(\ref{hoof6}) results in
					\begin{equation}
						w(k'')-w(k') \geq \epsilon/(n+1)^{l-1} > \epsilon/(n+1)^{l}.
					\end{equation}
					
				{\it Case 2:} If (\ref{dar4}) holds, from (\ref{dar6}), we have
					\begin{equation}
						|\n(k') \cap \C_{l+1}| < |\n(k'') \cap \C_{l+1}|,
					\label{dar8}
					\end{equation}
where the difference between the two expressions on the opposite side of inequality (\ref{dar8}) is at least $1$. We also know from~(\ref{dar5}) that
					\begin{equation}
						\begin{array}{l}
							\min \big\{ w(k)\,|\, k\in \n(k') \cap \C_{l-1} \big\}\vspace{.05in}\\
							\leq \min \big\{ w(k)\,|\, k\in \n(k'') \cap \C_{l-1} \big\}. 
						\end{array}
					\label{dar9}
					\end{equation}
Using (\ref{dar8}) and (\ref{dar9}) in  (\ref{hoof6}), we have
					\begin{equation}
						w(k'') - w(k') \geq \frac{\epsilon}{(n+1)^{l-1}}\left(\frac{1}{n+1}\right) = \frac{\epsilon}{(n+1)^l}, \notag
					\end{equation}
which completes the proof.
	\end{proof}		
	Recall that we wish to show that for every $i \in \V \backslash \C_0$, and $j \in \V$:
	\begin{equation*}
		e_{ij} \in \e ~\Leftrightarrow~
			\begin{cases}
				w(i) + w(j) \geq \tp,\\
				and\\
				|w(i)-w(j)| \leq \tm.
			\end{cases}
		\end{equation*}
	
	Without loss of generality, assume next that $i \in \C_l$, $1 \leq l \leq m$, and $j \in \C_{l'}$, where $0 \leq l' \leq l$. We analyze the cases $l' \leq l -2$, $l' = l-1$, and $l' =l$ as follows.
	\begin{itemize}
		\item[1.] If $l' \leq l - 2$, we know from the defining property of the distance decomposition $(\C_0,\C_1,\ldots,\C_m)$ that $e_{ij} \not\in \e$. On the other hand, according to Lemma \ref{lem1}, $w(i)-w(j) > \tm$. Thus, the Condition (\ref{DT condition}) holds.
		\item[2.] If $l' = l-1$, we consider two possibilities: $e_{ij} \in \e$ and $e_{ij} \not\in \e$. If $e_{ij} \in \e$, from (\ref{hoof6}) we have 
			\begin{equation}
				w(i) \leq \tm + \Big( \min \big\{ w(k)\,|\, k\in \n(i) \cap \C_{l-1} \big\} \Big) \leq \tm + w(j). \notag
			\end{equation}
On the other hand, according to Lemma \ref{lem1}, we conclude that $w(i) + w(j) \geq \tp$. Thus, (\ref{DT condition}) holds. 

If $e_{ij}\not\in \e$, then $j \neq j',$ where
$$ j'\triangleq {\rm argmin} \big\{ w(k)\,|\, k\in \n(i)\}. \notag$$
If  $w(j') > w(j)$, then from Lemma \ref{lem2}, 
$$w(j') - w(j) > \frac{\epsilon}{(n+1)^{l-1}}.$$
 Thus, from the recursive relation (\ref{hoof6}), it is straightforward to show that $w(i) > w(j) +\tm$. As a result, Condition (\ref{DT condition}) is satisfied. The inequality $w(j') \leq w(j)$ is impossible, since otherwise from (\ref{dar10}) and $e_{ij'}\in \e$, one would have $e_{ij} \in \e$.
		\item[3.] If $l'=l$, then $e_{ij} \in \e$ according to (\ref{con2}). From Lemma \ref{lem1}, we deduce that both $w(i) + w(j) \geq \tp$ and $|w(i) - w(j)| \leq \tm$ are satisfied. Hence, Condition~(\ref{DT condition}) holds.
	\end{itemize}
	This completes the proof of Theorem~\ref{maint}.

\section{A Polynomial-time Algorithm for Identifying PT Graphs}
\label{Algorithm}
Having characterized PT graphs and assigned weights to a PT graph
given the thresholds $\tp$ and $\tm$, we are now ready to describe a polynomial-time
algorithm for checking if a given graph $\G(\V,\e)$ is PT or not. The algorithm 
produces a distance decomposition satisfying Conditions (i)-(iii) of Theorem~\ref{maint} for a PT graph which is
not a unit interval graph. If $\G$ is not a PT graph, the algorithm
finds a forbidden induced subgraph in $\G$ or shows that there does
not exist a distance decomposition satisfying Conditions (i)-(iii) of
Theorem~\ref{maint} in $\G$.  

We start by providing necessary definitions and concepts needed to analyze the
algorithm and then proceed to outline the polynomial-time algorithm
itself.

We begin by recalling the definition of chordal graphs, along with a basic characterization due to Fulkerson and Gross~\cite{fulkerson-gross} as well as Rose~\cite{rose}.
\begin{definition}
  A graph is \emph{chordal} if it has no induced cycle of length greater than $3$.
\end{definition}
\begin{definition}
  A \emph{simplicial vertex} in a graph $H$ is a vertex $v$ such that
  $\n(v)$ is a clique.
\end{definition}
\begin{lemma}[Fulkerson--Gross~\cite{fulkerson-gross}, Rose~\cite{rose}]\label{lem:simplicial}
  A graph $G$ is chordal if and only if every induced subgraph of $G$ has a simplicial vertex.
\end{lemma}
Lemma~\ref{lem:simplicial} implies that every PT graph is chordal.
\begin{lemma}\label{lem:chordal}
  If $G$ is a PT graph, then $G$ is chordal.
\end{lemma}
\begin{proof}
  Since every induced subgraph of a PT graph is a PT graph, it
  suffices, by Lemma~\ref{lem:simplicial}, to show that every PT graph
  has a simplicial vertex. Suppose $G$ is an $(\tp,\tm,w)$-PT graph.  Let $i$ be
  a vertex minimizing $w$. We claim that $i$ is a simplicial
  vertex. Let $j,k$ be any distinct vertices in $\n(i)$; we may assume
  that $w(k) \geq w(j)$.
  
  Since $w(i)$ is minimum among all vertices and $e_{ij}, e_{ik} \in E$,
  we have $w(j), w(k) \in [w(i), w(i) + \tm]$, so that $\sizeof{w(i) - w(j)} \leq \tm$.
  Since $w(j) + w(k) \geq w(j) + w(i) \geq \tp$, we have $e_{jk} \in E$, implying
  that $\n(i)$ is a clique.
\end{proof}
Next we recall the following forbidden subgraph characterization of
unit interval graphs, and introduce the related notion of \emph{semi-unit-interval graphs}.
\begin{lemma}[Roberts~\cite{roberts-unit}]
  A graph is unit interval if and only if it is chordal and contains no induced subgraphs isomorphic to the $K_{1,3}$, sun and net graphs shown in Figure~\ref{fig:k13netsun} ($K_{1,3}$ in Figure~\ref{fig:k13netsun}(a), sun graph in Figure~  \ref{fig:k13netsun}(b) and net graph in Figure~\ref{fig:k13netsun}(c)).
\end{lemma}
\begin{figure}[t]
        \begin{center}
				\subfigure[]{
         \centering\includegraphics[width=1in]{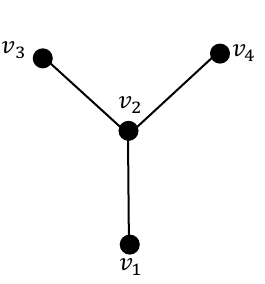}
         }	
        \subfigure[]{
         \centering\includegraphics[width=1in]{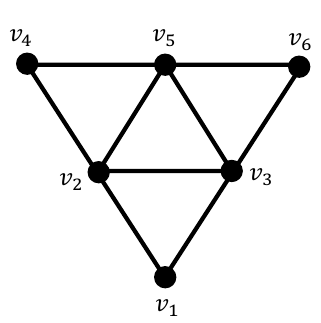}
         }	
				\subfigure[]{
         \centering\includegraphics[width=0.8in]{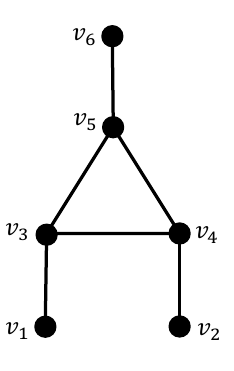}
         }
     
        \end{center}
        \caption{ Three forbidden induced subgraphs in unit interval graphs, (a) $K_{1,3}$, (b) a sun, (c) a net.}
				\label{fig:k13netsun}
				\vspace{-0.15in}
\end{figure}

\begin{definition}
\label{def1_Gr}
A graph $\G(\V,\e)$ is semi-unit-interval if it is chordal and has no
induced subgraph isomorphic to a net or a sun.
\end{definition}
\begin{lemma}
\label{lem2_Gr}
   If $\G(\V,\e)$ is a PT graph, then $G$ is semi-unit-interval.
\end{lemma}
\begin{proof}
  By Lemma~\ref{lem:chordal}, PT graphs are chordal. It remains to show
  that a PT graph has no induced subgraphs isomorphic to a sun or a net.
  Since every induced subgraph of a PT graph is a PT graph, it
  suffices to show that the sun and the net are not PT graphs.

  First, let $G$ be a graph isomorphic to the net; we show that $G$ is
  not a PT graph. Suppose to the contrary that $G$ is a PT graph, and
  let $(\C_0, \C_1, \ldots, \C_m)$ be a distance decomposition with
  the properties guaranteed by Theorem~\ref{maint}. Since $G$ is not a
  unit interval graph, $\C_0$ is nonempty.  Observe that all vertices
  of degree $1$ are pairwise incomparable in the vicinal preorder;
  likewise, all vertices of degree $3$ are pairwise incomparable in
  the vicinal preorder. Thus, $\C_0$ contains at most one vertex of
  degree $1$ and at most one vertex of degree $3$. Using the symmetry
  of $G$, it is straightforward (if slightly tedious) to check all
  possible such choices of $\C_0$, and to observe that for each
  possible choice, one of the sets $\C_l$ for $l > 0$ is not a
  clique, contradicting our choice of the distance decomposition
  to satisfy the properties guaranteed by Theorem~\ref{maint}.

  Next, let $G$ be a graph isomorphic to the sun; we show that $G$ is
  not a PT graph. Again, suppose to the contrary that $G$ is a PT
  graph, and let $(\C_0, \C_1, \ldots, \C_m)$ be a distance
  decomposition with the properties guaranteed by
  Theorem~\ref{maint}. As before, since $G$ is not a unit
  interval graph, we see that $\C_0$ is nonempty, and as before,
  all vertices of degree $2$ are pairwise incomparable in the
  vicinal preorder, as are all vertices of degree $4$. Consequently, 
  $\C_0$ contains at most one vertex of degree $2$ and at most
  one vertex of degree $4$. It is again straightforward but tedious
  to check that each possible choice of $\C_0$ satisfying these constraints
  leads to one of the sets $\C_l$ for $l > 0$ failing to be a clique,
  yielding a contradiction.
\end{proof}
The results of Section~\ref{DTG} imply that any PT graph $G$
admits a partition $(V_T, V_U)$ of its vertices such that
the subgraph induced by $V_T$ is a threshold graph and
the subgraph induced by $V_U$ is a unit interval graph.
Seeking a converse, we look for conditions on a vertex partition
$(V_T, V_U)$ which guarantee that the graph being partitioned
is a PT graph. The relevant notion turns out to be an \emph{admissible}
partition.
\begin{definition}
\label{def3_Gr}
 Let $\G(\V,\e)$ be a semi-unit-interval graph, and let $(V_T, V_U)$ be a partition of $\V$.  We say that $(V_T, V_U)$ is \emph{admissible} if all the following conditions hold:\\
\begin{enumerate}[(1)]
  \item No two vertices  of $V_T$   are incomparable in the vicinal preorder,
  \item For every $i \in V_T$, the set $\n(i) \cap V_U$  is a clique, and
  \item There are no induced subgraphs that have any of the induced \emph{colorings} shown in Figure~\ref{fig:bull-k13}.
\begin{figure}[t]
        \begin{center}
        \subfigure[]{
         \centering\includegraphics[width=1in]{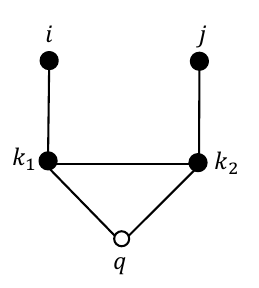}
         }	
				\subfigure[]{
         \centering\includegraphics[width=0.9in]{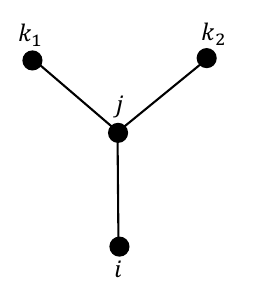}
         }
     \subfigure[]{
         \centering\includegraphics[width=0.9in]{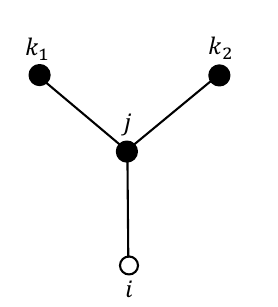}
         }
        \end{center}
        \caption{Forbidden induced colorings of a bull (a) and $K_{1,3}$ ((b), (c)) in an admissible partition,  where $\circ$ denotes a vertex in $V_T$ and $\bullet$ denotes a vertex in $V_U$. We refer to $\circ$ and $\bullet$ as colors.}
				\label{fig:bull-k13}
				\vspace{-0.15in}
\end{figure}
\end{enumerate}
\end{definition}
For any vertex set $X$, let $G[X]$ denote the subgraph of $G$ induced
by the vertex set $X$.  Observe that if $(V_T, V_U)$ is an admissible
partition of $G$, then Condition~1 immediately implies that $G[V_T]$
is a threshold graph. Similarly, Condition 3 implies that $G[V_U]$ is
a unit interval graph, since it implies that $G[V_U]$ has no induced
subgraph isomorphic to $K_{1,3}$, and the other forbidden induced
subgraphs for unit interval graphs are already forbidden in $G$ due
to $G$ being a semi-unit-interval graph.

One can think of Condition~3, in particular, as a version of a
``forbidden induced subgraphs'' condition: while we are not able to
characterize PT graphs by their forbidden induced subgraphs, the
following theorem characterizes them as being the graphs that
admit a $2$-coloring which omits a set of induced
\emph{colorings}. In fact, the first two conditions can also be
reformulated, with some effort, as forbidding certain colorings of a
set of induced subgraphs, but we have chosen to state them in a more direct way.
\begin{theorem}
  Let $\G(\V,\e)$ be a semi-unit-interval graph.  The graph $G$ is
  paired threshold if and only if it has an admissible partition.
  \label{Th:p-admissible}
\end{theorem}
We apply Theorem~\ref{Th:p-admissible} to devise an algorithm for determining
whether a graph is paired threshold, shown in Algorithm~\ref{alg:DT}. The algorithm 
requires one more definition, a specialized version
of admissible partitions.
\begin{definition}
  Let $G$ be a graph, and let $v \in V$. A vertex partition $(V_T, V_U)$
  is \emph{$v$-admissible} if:
  \begin{itemize}
  \item $(V_T, V_U)$ is admissible,
  \item $v \in V_T$, and
  \item Among the vertices of $V_T$, the vertex $v$ is maximal in the
    vicinal preorder.
  \end{itemize}
\end{definition}
\begin{algorithm}
  \caption{Determine whether a graph $G$ is paired threshold,
    and if so, return an admissible partition $(V_T, V_U)$.}
  \label{alg:DT}
  \begin{algorithmic}
    \IF{$G$ is not semi-unit-interval}
    \STATE{Return ``False''.}
    \ENDIF
    \IF{$G$ is unit interval}
    \STATE{Return the partition $(\emptyset, V)$.}
    \ENDIF
    \FOR{$v \in V$}
    \IF{there is a $v$-admissible partition $(V_T, V_U)$}
    \STATE{Return the partition $(V_T, V_U)$}
    \ENDIF
    \ENDFOR
    \STATE{Return ``False''.}
  \end{algorithmic}
\end{algorithm}
It is known that chordality testing for a graph with $n$ vertices and
$m$ edges can be carried out in $O(n+m) = O(n^2)$
time~\cite{rose-tarjan-lueker, tarjan-yannakakis, TY-addendum}, and as
there are only two other forbidden induced subgraphs for a graph to be
semi-unit-interval, each of which has $6$ vertices, we can test
whether a graph is semi-unit-interval in $O(n^6)$ time, simply by
first checking whether $G$ is chordal, and, if so, testing each
possible set of $6$ vertices for the remaining forbidden induced
subgraphs. Furthermore, one can determine in linear time whether
  a given graph is a unit-interval graph~\cite{linear-unit-interval}.

In order to determine whether a $v$-admissible partition exists (for a specified $v$), 
we will produce a 2SAT instance whose satisfying
solutions correspond to $v$-admissible partitions of $G$. It is known
that a 2SAT instance with $t$ clauses can be solved in $O(t)$
time~\cite{EIS-2sat}, and our construction will produce a
polynomially-sized 2SAT instance in polynomial time, so this yields a
polynomial-time algorithm for checking whether a $v$-admissible
partition exists. A more detailed complexity analysis will be given
at the end of the section.

\begin{lemma}  
\label{lem5}
  If $(V_T, V_U )$ is a $v$-admissible partition, then for every  $i \in V_T$, we have
  $\n(i) \subseteq \n (v) \cup \{v\}$.
\end{lemma}
\begin{proof}
  This follows immediately from the facts that the vicinal preorder is
  total on $V_T$ and that $v$ is maximal among the vertices of $V_T$ in the
  vicinal preorder.
\end{proof}

\begin{proof}[Proof of Theorem~\ref{Th:p-admissible}]
  First, observe that if $G$ is a semi-unit-interval graph, then it is
  unit interval if and only if it has no induced copy of $K_{1,3}$.
  Equivalently, under the hypothesis that $G$ is semi-unit-interval,
  $G$ is unit interval if and only if the partition $(\emptyset, V)$
  is an admissible partition. Thus, for the remainder of the proof, we may
  assume for the forward direction that $G$ is not a unit interval graph,
  and for the reverse direction that $V_T \neq \emptyset$.
  
  Let $G$ be a paired threshold graph that is not a unit
  interval graph; we show that it has an admissible partition. Let
  $w$ be a weight function 
  $(\C_0, \C_1, \ldots, \C_m)$ be a distance decomposition with the
  properties guaranteed by Theorem~\ref{maint}. Let $V_T = \C_0$ and
  let $V_U = C_1 \cup \cdots \cup \C_m$. We claim that $(V_T, V_U)$ is
  an admissible partition. Conditions (1) and (2) follow immediately
  from the properties guaranteed by Theorem~\ref{maint}.  To
  prove that Condition~(3) holds, recall the definition of the preorder $\prel$ on $\C_l$ given in
  (\ref{hoof}):
  \begin{equation*}
    i \,\prel\, j ~\Leftrightarrow~
    \begin{cases}
      \n(j)\cap \C_{l-1}  \subseteq \n(i) \cap \C_{l-1},\\
      and\\
      \n(i)\cap \C_{l+1}  \subseteq \n(j) \cap \C_{l+1},
    \end{cases}
  \end{equation*}
  We now show none of the forbidden induced colorings appear in $G$:
  \begin{itemize}
  \item Suppose that $X$ is the vertex set of a forbidden bull, and
    let $i$ and $j$ be the vertices of degree $1$.  If $i$ or $j$ is
    adjacent to some vertex of $V_T$, then the set of neighbors of
    $V_T$ does not form a clique, which is a contradiction to Theorem~\ref{maint}.  
    If neither $i$ nor $j$ is adjacent to a vertex of
    $V_T$, then both $i$ and $j$ have distance exactly $2$ from $V_T$;
    by Theorem~\ref{maint} this implies that $i$ and $j$ should be
    adjacent, which is not the case.

  \item Since $G[V_U]$ is a unit interval graph and $K_{1,3}$ is a
    forbidden induced subgraph for unit interval graphs, there cannot
    be any $K_{1,3}$ for which all vertices lie in $V_U$.

  \item Suppose that $X$ is the vertex set of a $K_{1,3}$ with all
    vertices in $V_U$ except for a single leaf vertex $i \in V_T$.
    Let $j$ be the center vertex of the $K_{1,3}$. As $j$ has a neighbor
    in $V_T = \C_0$, we have $j \in \C_1$. Letting $k_1$ and $k_2$ be
    the other leaves of the $K_{1,3}$, we see that since each of $k_1$
    and $k_2$ is a vertex of $V_U$ adjacent to the vertex $j$ in $\C_1$,
    we must have $\{k_1, k_2\} \subset \C_1 \cup \C_2$. Furthermore,
    since each $\C_i$ is a clique, the vertices $k_1$ and $k_2$ cannot
    both lie in $\C_1$, nor can they both lie in $\C_2$. Thus, we may
    assume that $k_1 \in \C_1$ and $k_2 \in \C_2$. Now since $i \in \n(j) \cap \C_0$
    but $i \notin \n(k_1) \cap \C_0$, we have $\n(j) \cap \C_0 \not\subset \n(k_1) \cap \C_0$,
    and since $k_2 \in \n(j) \cap \C_2$ but $k_2 \notin \n(k_1) \cap \C_2$, we have
    $\n(j) \cap \C_2 \not\subset \n(k_1) \cap \C_2$. This implies that neither
    $j \,\pretty_1\, k_1$ nor $k_1 \,\pretty_1\, j$ hold, which contradicts
    the property that the preorder $\pretty_1$ is total on $\C_1$.
\end{itemize}
Now, let $(V_T , V_U)$ be an admissible partition of $V$ with
$V_T \neq \emptyset$, let $\C_0 = V_T$, and let
$(\C_0, \C_1, \ldots, \C_m)$ be the resulting distance
decomposition. We will verify that the distance decomposition
satisfies Conditions (i)-(iii) of Theorem~\ref{maint}, which implies
that $G$ is a PT graph.

  Condition~(i) of Theorem~\ref{maint} follows immediately
  from Condition~(1) of the definition of an admissible partition,
  since $\C_0 = V_T$. 

  Let $\C_0  = V_T$ and, for $l \geq 1$, define $\C_l$ as in~(\ref{eq:cliques}), i.e.,
  \begin{equation}
    \C_l = \left\{ i \in \V \,{\Big{|}}\,  \dist(i,j)_{j \in \C_{0}} = l\right\} \nonumber.
  \end{equation}
  Next we establish Condition~(ii) of Theorem~\ref{maint}, which
  states that each set $\C_l$ for $l > 0$ is a clique. First we
  argue that $\C_1$ is a clique. Let $i$ be a vertex of $\C_0$ which
  is maximal in the vicinal preorder. Every vertex of $\C_1$ is adjacent
  to a vertex of $\C_0$ and thus, by the maximality of $i$, every vertex
  of $\C_1$ is adjacent to $i$. Thus, $\C_1 \subset \n(i) \cap V_U$. Now, applying Condition~(2) to $v$, we
  see that $\n(i) \cap V_U$ is a clique, hence $\C_1$ is a clique.
  
Assuming that $\C_l$ is a clique, we now show that $\C_{l+1}$ is
  also a clique. Let $i,j \in \C_{l+1}$ and suppose that $i$ and $j$
are nonadjacent.  Each of the vertices $i$ and $j$ have at least one
neighbor in $\C_l$.

  {\bf{Case 1}}: The vertices $i$ and $j$ have a common neighbor
  $k \in \C_l$. The vertex $k$ has a neighbor $q \in \C_{l-1}$; now,
  ${i,j,k,q}$ is an induced $K_{1,3}$ subgraph with the center
  $k \in V_U$ and at least two leaves $i$, $j$ in $V_U$; but this
  configuration is forbidden.

  {\bf{Case 2}}: The vertices $i$ and $j$ have no common neighbor in
  $\C_l$.  Let $i' \in \n(i) \cap \C_l$ and let
  $j' \in \n(j) \cup \C_l$. Since $\C_l$ is a clique,
  $e_{i'j'} \in E$.  If $l = 1$, then let $v$ be a vertex of $\C_0$
  which is maximal in the vicinal preorder; we have that
  $v \in \n (i') \cap \n (j')$, since each vertex of $\C_1$ is
  adjacent to a vertex of $\C_0$, and the vicinal preorder is total
  on $\C_0$. Now ${v, i, j, i', j'}$ induces a forbidden coloring
  of vertices of a bull.  If $l > 1$, let
  $k \in \n (i') \cap \C_{l-1}$. If $e_{kj'} \notin E$, then
  ${i', j', i, k}$ is an induced $K_{1,3}$ with $i'$ as its center and
  all its vertices in $V_U$, which is forbidden.  If $e_{kj'} \in E$,
  then since $l - 1 \geq 1$, we see that $k$ has some neighbor
  $q \in \C_{l-2}$.  Since $q$ cannot be adjacent to any of
  ${i', j', i, j}$, we see that ${i, j, i', j', k, q}$ induce a net.
  This contradicts the assumption that $G$ is semi-unit-interval.

  Finally, we verify Condition~(iii) of Theorem~\ref{maint}. We must show that $\prel$ is a total preorder on each
  $l$.  Let $i, j \in \C_l$ and suppose to the contrary that
  $i, j$ are incomparable in $\prel$. There are four possibilities (in
  fact, only two possibilities, up to symmetry), each of which may be eliminated as follows.

  {\bf{Case 1:}} One has $\n(i) \cap \C_{l-1} \not\subseteq \n (j) \cap \C_{l-1}$
  and $\n (j) \cap \C_{l-1} \not\subseteq \n (i) \cap \C_{l-1}$.  In this case,
  there exist $i', j' \in \C_{l-1} $ with
  $i' \in \n (i) \backslash \n (j)$ and
  $j' \in \n (j) \backslash \n (i)$.  If $l > 1$, then since
  $\C_{l-1}$ and $\C_{l}$ are cliques, this implies that $ii'j'j$
  induces a $C_4$ in $G$, contradicting the assumption that $G$ is
  chordal.  If $l = 1$, then this implies $i'$ and $j'$ are vertices
  of $V_T$ that are incomparable in the vicinal preorder,
  contradicting the assumption that $(V_T, V_U)$ is admissible.

  {\bf{Case 2}}: One has $\n (i) \cap \C_{l+1} \not\subseteq \n(j) \cap \C_{l+1}$
  and $\n (j) \cap \C_{l+1} \not\subseteq \n (i) \cap \C_{l+1}$.  By symmetry, this
  is covered by Case 1.

  {\bf{Case 3}}: One has $\n (i) \cap \C_{l-1} \not\subseteq \n (j) \cap \C_{l-1}$
  and $\n (i) \cap \C_{l+1} \not\subseteq \n (j) \cap \C_{l+1}$. Take
  $i_1 \in (\n (i) \cap \C_{l+1} ) \backslash \n (j)$ and
  $i_{2} \in (\n (i) \cap \C_{l-1} ) \backslash \n (j)$.  Observe that
  $e_{i_1 i_{2}} \notin E$, since if this edge were present, then $i_1$ would
    have distance at most $l$ to some vertex of $V_T$, contradicting $i_1 \in \C_{l+1}$
  every vertex of $V_T$. Hence ${i, j, i_{2} , i_1}$ induce a $K_{1,3}$ subgraph 
  in $G$, with only the vertex $i_{2}$ possibly belonging to $V_T$;
  this is a forbidden induced partition.

  {\bf{Case 4}}: One has $\n (j) \cap \C_{l-1} \not\subseteq \n (i) \cap \C_{l-1}$
  and $\n (j) \cap \C_{l+1} \not\subseteq \n (i) \cap \C_{l+1}$.  By symmetry, this
  is covered by Case 3.
\end{proof}
\begin{corollary}\label{cor:w}
  If $(V_T, V_U)$ is a $v$-admissible partition and $i$ is a vertex with $\n(i) \not\subseteq \n(v) \cup \{v\}$,
  then $i \in V_U$.
\end{corollary}
Let $W = \{i \in V(G) \st \n(i) \subseteq \n(v) \cup \{v\}\}$. By Corollary~\ref{cor:w}, we have
$V_T \subseteq W$ for any $v$-admissible partition $(V_T, V_U)$.
                                                                                                            
\begin{lemma}
\label{lem7_Gr}
Let $v \in V$  and let $(V_T, V_U)$ be a partition of $V$ such that
\begin{enumerate}
 \item One has $v \in V_T$;
\item  All vertices  of $V \backslash W$ are in $V_U$; and
\item  All vertices in $V_U$ that are adjacent to $v$ form a clique.
\end{enumerate}
If $(V_T, V_U)$ has one of the forbidden induced colorings in Figure~\ref{fig:bull-k13}, then either $G$ has a forbidden bull in which $v \in V_T$,  or $G$ has a forbidden induced  $K_{1,3}$  in which some vertex of $V \backslash W$ is the center vertex.
\end{lemma}
\begin{proof}  
First suppose that $S$ is the vertex set of an induced forbidden bull, with vertices labeled as shown in Figure~\ref{bull2}. 
	\begin{figure}[t]
		\centering
		\includegraphics[width=1.2in]{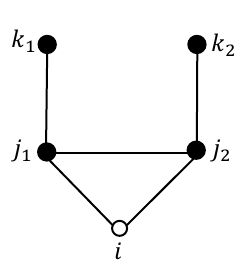}
		\caption{A partitioning of the vertices of a bull that is used in the proof of Lemma~\ref{lem7_Gr}, where $i \in V_T$ and $j_1, j_2, k_1 ,k_2 \in V_U$. }
	\label{bull2}
	\end{figure}
Since $i \in V_T$, we have $i \in W$. Therefore, $\{j_1, j_2 \} \subseteq \n(v)$.  Since $e_{j_1 k_2}, e_{j_2 k_1}  \notin E$ and  since the vertices in $V_U$ that are incident to $v$ form a clique, we see that $k_1, k_2  \notin  \n (v)$. Therefore, $(S \backslash \{i\}) \cup \{v\}$ also induces a forbidden  bull.

Now suppose that $S$ is the vertex  set of a forbidden induced coloring of $K_{1,3}$. Since the  vertices in $V_U$ that are in $\n(v)$ form a clique, at most one leaf vertex of $S$ lies in $\n(v) \cap V_U$.  In particular, the  center  vertex  of $S$ has a neighbor outside $\n (v) \cup \{v\}$, which implies that the center vertex does not lie in $W$, by Corollary~\ref{cor:w}.
\end{proof}
We are now in a position to define the 2SAT instance modeling the $v$-admissible partition problem.
\begin{definition}
 Given a semi-unit-interval  graph  $G$ and a vertex $v \in V$, we define a 2SAT instance as follows.
 \begin{enumerate}[(i)]
 \item For  each $i \in V$,  we define a variable  $x_i$, with the intended interpretation that $x_i$ is true if and only if $i \in V_T$ in the partition;
 \item We add a clause $(x_v  \vee x_v )$,  and  for each  $i \in V \backslash W$, we add a clause $(\neg x_i  \vee \neg x_i)$;
 \item For each nonadjacent pair of vertices  $i, j \in \n(v)$, we add a clause $(x_i  \vee x_j)$;
 \item  For  each pair  of vertices  $i, j$ that are incomparable in the vicinal  preorder, we add a clause $( \neg x_i  \vee \neg x_j )$;
 \item For every pair of vertices  $i, j$ that are the leaves of some induced  bull with $v$ as the degree-2 vertex, we add a clause $(x_i  \vee x_j )$;
 \item For every copy of $K_{1,3}$ with the center vertex $k \in V  \backslash W$ with leaves $i, j, q$,
we add three  clauses $(x_i \vee x_j )$, $(x_i \vee x_q )$, $(x_j  \vee x_q )$.
\end{enumerate}
\end{definition}
\begin{theorem} 
For  any semi-unit-interval graph $G$ and any $v \in V$,  $G$ has a $v$-admissible partition if and only if the associated  2SAT instance  is satisfiable.
\end{theorem}
\begin{proof}
  First, suppose that $G$ has a $v$-admissible partition $(V_T, V_U )$. Consider the 2SAT assignment obtained by letting
  $x_i$ be true if and only if $i \in V_T$. We verify that all clauses of the 2SAT instance are satisfied:

\begin{itemize}
\item By Corollary~\ref{cor:w}, all clauses added in step (ii) are satisfied.
\item Since in a $v$-admissible partition, the vertices in $V_U$ that are adjacent to $v$ form a clique, all
clauses added in step (iii) are satisfied.
\item Since in a $v$-admissible partition the vicinal preorder is total  on $V_T$, all clauses added in step (iv) are satisfied.
\item Since a $v$-admissible partition omits the forbidden induced subgraphs of Definition~\ref{def3_Gr}, all clauses added in steps (v) and (vi) are satisfied.
\end{itemize}

On the other hand, suppose that the 2SAT instance is
  satisfiable. Let $(V_T, V_U )$ be the partition obtained by putting
  $i \in V_T$ if and only if $x_i$ is true; we will prove that
  $(V_T, V_U)$ is a $v$-admissible partition.  First, observe that the
  clauses added in step~(ii) guarantee that $v \in V_T$ and that only
  vertices of $W$ can be in $V_T$, so $v$ is maximal among the
  vertices of $T$ in the vicinal preorder. Hence, if $(V_T, V_U)$ is
  admissible, then it is $v$-admissible.

To show that $(V_T, V_U)$ is admissible, we verify the conditions of Definition~\ref{def3_Gr}. Conditions (1) and (2) of Definition~\ref{def3_Gr} are easy to verify:
\begin{enumerate}[(1)]
\item No two vertices  of $V_T$ are incomparable in the vicinal preorder, since this would violate a clause added in step (iv).

\item If for some $i \in V_T$ the set $\n (i) \cap V_U$  is not a clique, then by the maximality of $\n (v)$, we also have that $\n (v) \cap V_U$  is not a clique, which would violate a clause added in step (iii).
\end{enumerate}

To verify Condition~(3) of Definition~\ref{def3_Gr}, we first observe
that satisfying the clauses added in steps (ii) and (iii) implies that
$(V_T, V_U )$ satisfies the hypothesis of Lemma~\ref{lem7_Gr}. Hence,
if $G$ has a forbidden induced bull as described in Definition~\ref{def3_Gr}, then
by Lemma \ref{lem7_Gr}, we can find such a forbidden induced coloring
with $v$ as the vertex of degree $2$, which violates a clause added in
step~(v).  Likewise, if $G$ has an induced $K_{1,3}$ with one of the
forbidden colorings, then by Lemma~\ref{lem7_Gr}, we can find some
forbidden $K_{1,3}$ whose center lies in $V \backslash W$, violating
some clause added in step~(vi). Thus, $(V_T, V_U)$ is admissible, which
  implies, by our earlier argument, that it is $v$-admissible.
\end{proof}

To complete the proof that Algorithm~\ref{alg:DT} runs in polynomial time, 
observe that the desired 2SAT instance can be
constructed in time $O(n^5)$, and has at most $O(n^5)$ clauses. Since
a 2SAT instance with $t$ clauses can be solved in $O(t)$ time~\cite{EIS-2sat},
this implies that one can check whether a $v$-admissible partition
exists (and construct one, if so) in time $O(n^5)$. With $n$ possible choices
for the vertex $v$, one can check whether an admissible partition exists
in time $O(n^6)$, so Algorithm~\ref{alg:DT} takes time $O(n^6)$ in total.

\section{Intersection Number, Diameter and Clustering Coefficient of PT Graphs} \label{sec:social}

Several measures for assessing the quality of graph models for social, economic, and biological networks include the vertex degree distribution, excluded subgraphs and network motifs, the graph diameter, intersection number and clustering coefficient. The vertex degree distribution describes the number of vertices of each degree in the graph, and is usually assumed to follow a power law~\cite{girvan2002community}. The diameter of a graph is the length of the longest shortest path between any two vertices of a graph, and it is known to be a small constant for many known social and biological networks~\cite{jackson2008social}. The intersection number of the graph describes latent network features~\cite{dau2017latent}, while a large clustering coefficient ensures that the model correctly contains a large number of triangles known to be biological and social network motifs, as described below. 

In his comprehensive study of social network motifs, Ugander~\cite{Ugander} determined the frequency of induced subgraphs with three and four vertices in a large cohort of interaction and friendship networks. In addition to showing that $K_3$ and $K_4$ cliques are the most prominent network motifs (i.e., subgraphs that appear with significantly higher frequency than predicted by some random model), Ugander also established the existence of \emph{anti-motifs} (e.g., highly infrequent induced subgraphs or forbidden induced subgraphs). For example, cycles of length four ($C_4$) represent the least likely induced subgraphs in social networks. Using the properties of PT graphs established in the previous section, it is straightforward to determine the structure of some of their forbidden induced subgraphs. In addition to avoiding induced cycles of length exceeding $3$, PT graphs may also be easily shown to avoid the subgraphs depicted in Figure~\ref{fig:forbiddensubgraphs}. Subgraph avoidance is, in general, is most easily established by showing that PT graphs belong to a larger family of graphs with well-characterized forbidden induced subgraphs. For instance, given that PT graphs are chordal, the forbidden induced subgraphs of chordal graphs are automatically inherited by PT graphs. As another example, with regards to the subgraph \ref{fig:forbiddensubgraphs}(a), it can be easily seen that for any choice of vertices satisfying Condition (i) of Theorem~\ref{maint}, Condition (ii) of Theorem~\ref{maint} is not met and therefore, there is no a distance decomposition $(\C_0,\C_1,\ldots,\C_m)$ satisfying the conditions of Theorem~\ref{maint}. Unfortunately, it appears difficult to characterize \emph{all forbidden subgraphs} of PT graph. 

In what follows, we provide a brief analysis of (a) the diameter of PT graphs, capturing relevant connectivity properties of networks; (b) the intersection number of PT graphs, which is of relevance for latent feature modeling and inference in social networks~\cite{babis,dau2016,IntSocial2012}; and (c) the clustering coefficient, providing a normalized count of the number of triangles in the graphs.

\begin{figure}[t]
        \begin{center}
        \subfigure []{
         \centering\includegraphics[width=1.8in]{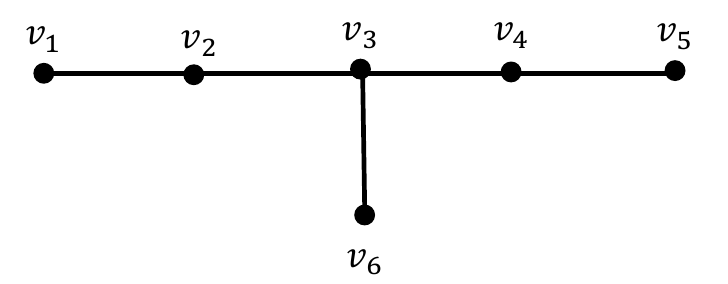}
         }	
				\subfigure []{
         \centering\includegraphics[width=1.4in]{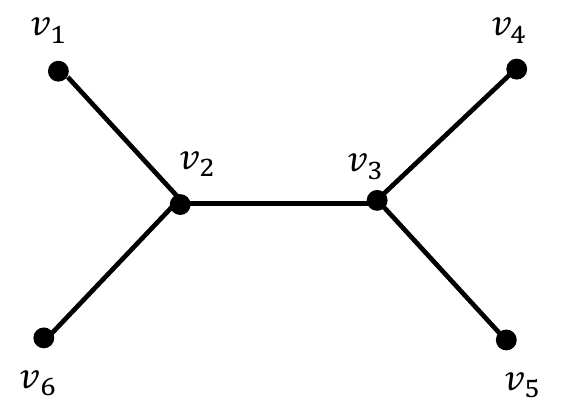}
         }
      \subfigure []{
         \centering\includegraphics[width=1.7in]{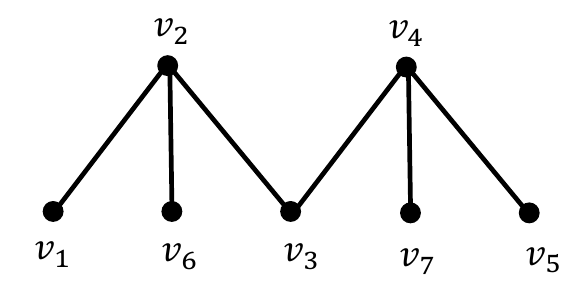}
         }
          \subfigure []{
         \centering\includegraphics[width=0.8in]{C4.pdf}
         }
        \subfigure []{
         \centering\includegraphics[width=1in]{sun.pdf}
         }  
				\subfigure []{
         \centering\includegraphics[width=0.85in]{net.pdf}
         }
        \end{center}
        \caption{Some forbidden induced subgraphs in PT graphs.}
				
				\label{fig:forbiddensubgraphs}
\end{figure}	
\subsection{The Diameter of a PT Graph}

The diameter of most social networks is a slowly growing function of the network size~\cite{broder2000graph, SmallWorld, barabasi1999diamiter}: in~\cite{bollobas2004diameter}, it was shown that preferential attachment graphs have diameters of size (sub)logarithmic in the number of vertices. The Small World phenomena~\cite{SmallWorld} suggests that the diameter of the underlying networks is close to six. In what follows, we investigate the diameter of PT graphs and determine under which conditions it matches the values observed in real social networks.

Denote the diameter of a connected PT graph by $D(G)$. Using the decomposition theorem for PT graphs, we can prove the following claim. 
\begin{theorem}
  Let $\G(\V,\e)$ be a connected PT graph with more than one vertex that is not unit interval. Let $(\C_0,\C_1,\cdots,\C_m)$ be a distance decomposition of $\G$ satisfying the conditions of Theorem \ref{maint}.
  If $m \geq 1$, then $D(G) = m+\lambda,$ where $\lambda \in \{0,1\}$.
\end{theorem}
\begin{proof}
  Clearly $D(G) \geq m$, since a vertex in $\C_m$ and a vertex in
  $\C_0$ have distance at least $m$. Thus, it suffices to show that
  $D(G) \leq m+1$.

  First we claim that for all $i,j \in \C_0$, we have ${\rm{ dist}}(i,j) \leq
  2$. 
  Choose $i$ and $j$ so that $i\,\prez\,j$. If $i$ and $j$ lie in the same component of 
  $\G[\C_0]$, then ${\rm{dist}}(i,j)\leq 2$ because a connected subgraph of a threshold graph with more than one vertex  is also a threshold graph. Using the recursive construction for threshold graphs, one can easily verify that in connected threshold graphs the diameter is at most two. 
  If $i$ and $j$ do not lie in the same component of 
  $\G[\C_0]$, since $i\,\prez\,j$, $i$ has to be an isolated vertex in $\G[\C_0]$.
  Since $\G$ is connected and $i\,\prez\,j$, then $i$ has a neighbor $k \in \C_1$, 
  where $k \in \n(j)$, and thus  ${\rm{dist}}(i,j)\leq 2$.

  Next we claim that for all $i \in \C_0$ and $j \in \C_l$, where $1
  \leq l \leq m$, we have ${\rm{ dist}}(i,j) \leq l+1$. Let $P$ be a path with
  $l-1$ edges from $j$ to some vertex $k \in \C_1$. Such a path
  necessarily exists, since for each $r$, any vertex in $\C_r$ has a
  neighbor in $\C_{r-1}$. If $i$ has some neighbor $q \in \C_1$,
  then $Pqi$ (or $Pi$ if $q=i$) is a $j,i$-path of length at most
  $l+1$.  Otherwise, let $q$ be a $\prez$-maximal vertex of $\C_0$;
  since $\G$ is connected, we have $i \in \n(q)$, so that $Pqi$ is
  again a path of length at most $l+1$.

  Finally, we claim that if $i \in \C_l$ and $j \in \C_r$ where $r
  \geq l$, then ${\rm{ dist}}(i,j) \leq (r-l)+1$.  Let $P$ be a path with $r-l$
  edges from $j$ to a vertex $k \in \C_l$. If $k \neq i$, then $Pi$ is
  a $j,i$-path of length $r-l+1$.

  In all cases, we have ${\rm{ dist}}(i,j) \leq m+1$.
\end{proof}

Since the diameter of a PT graph with distance decomposition $(\C_0,\C_1,\cdots,\C_m)$ is at most $m+1$, the question arises whether or not a given PT graph has a decomposition with $m \leq 5$. 

To answer this question, we use the decomposition algorithm described in the previous section. We know that in a $v$-admissible partition, the vertices that possibly lie in $V_T$ are the vertices at distance at most $2$ from $v$. In particular, if $m \geq 2$ and there is a vertex at distance greater than $m$ from $v$, then that vertex has to be in $V_U$ in any $v$-admissible partition, and will therefore be in a clique at distance $m+1$ from the threshold graph. Conversely, any vertex at distance at least $m+1$ from the threshold graph is also at distance $m+1$ from $v$.

So, for $m \geq 2$, there is a partition with at most $m$ layers in the unit-interval graph if and only if there is some vertex $v$ such that (1) every vertex is within distance $m$ of $v$, and (2) the graph has a $v$-admissible partition.

For the special case $m=1$, it is no longer necessary that every vertex is within distance $1$ of $v$, but the only way this is possible is if every vertex at distance $2$ from $v$ is in $V_T$. These vertices are isolated vertices in the threshold graph. Therefore, for each such vertex, we can add ($x_i \vee x_i$) as an additional constraint to the 2SAT problem and search for a vertex $v$ such that the modified 2SAT problem has a solution. This would produce the desired decomposition.

For the special case $m=0$, one only needs to check whether the graph is a threshold graph without isolates, which is straightforward to do, and as already mentioned, such graphs have diameter at most $2$.

\subsection{ Intersection Numbers of PT Graphs} \label{Intersection number}

We start by providing relevant definitions regarding intersection graphs and intersection representations~\cite{bookIntersection}.
\begin{definition}
Let $F=\{S_1,\ldots,S_n\}$ be a family of arbitrary sets (possibly with repetition). The intersection graph associated with $F$ is an undirected graph with  vertex set $F$ and the property that $S_i$ is adjacent to $S_j$ if and only if $ i\neq j$ and $S_i \cap S_j \neq \emptyset$. 
\end{definition}
We note that every graph can be represented as an intersection graph~\cite{Intersection}. 
\begin{definition}
The intersection number of a graph $\G(\V,\e)$ is the  cardinality of a minimal set $S$ for which $G$ is the intersection graph of a family of subsets of $S$. The intersection number of $\G$ is denoted by $\iota(G)$.
\end{definition}

Equivalently, the intersection number equals the smallest number of
cliques needed to cover all of the edges of $\G$ \cite{Gross,Erdos}. A
set of cliques with this property is known as an {\emph{edge clique cover}}.  In fact, an edge clique cover of $G$ is any family
$Q=\{Q_1, \cdots, Q_k\}$ of complete subgraphs of $G$ such that every
edge of $G$ is in at least one of $E(Q_1), \cdots, E(Q_k)$,
i.e. $e_{ij} \in E(G)$ implies that $e_{ij} \in \cup_{n=1}^k E(Q_n)$
\cite{bookIntersection}.

Scheinerman and Trenk~\cite{Scheinermann} gave an algorithm to compute
the intersection number of chordal graphs in polynomial time. Since PT graphs are chordal,
it is possible to apply the Scheinerman--Trenk algorithm to compute
the intersection number of PT graphs. In this section, however, we present an
explicit formula for the intersection number of PT graphs.
\begin{theorem}\label{thm:nplus}
  Let $\G$ be a $(\tp,\tm,w)$-PT graph, and let $1, \ldots, n$ be the
  vertices of $\G$, ordered so that $w(1) \leq w(2) \leq \cdots \leq
  w(n)$.  For each $i \in \{1, \ldots, n\}$, let $\n_i^+ = \{j>i \st
  e_{ij} \in E\}$.  If \[S = \{i \in V(\G) \st \text{$\n_i^+$ is
    nonempty and $\{i\} \cup \n_i^+ \not\subseteq \n_{i-1}^+$}\},\] then $\iota(G) =
  \sizeof{S}$.
\end{theorem}
\begin{proof}
  For each $i \in S$, let $C_i = \{i\} \cup \n_i^+$. We claim that
  $\{C_i\}_{i \in S}$ is an edge clique cover of $\G$. Let $e_{jk}$ be
  any edge of $\G$, with $j < k$, and let $i$ be the largest element
  of $S$ satisfying $i \leq j$. Such an element must exist, since if
  $\min S > j$, then $d(i) = 0$ for all $i \leq j$, contradicting the
  existence of the edge $e_{jk}$. Since $k \in \n^+_j$,  the definition of $S$ implies that if $i<j$, then
  $\{j\} \cup \n_{j}^+ \subseteq \n_{j-1}^+ \subseteq \cdots \subseteq \n_{i}^+$. Thus, $e_{jk} \in E(C_i)$. If $i=j$,  then  
	$\{j\} \cup \n_{j}^+ = \{i\} \cup \n_{i}^+$ and hence, 
  $e_{jk} \in E(C_i)$. This implies that $\iota(G) \leq \sizeof{S}$.

  To show that $\iota(G) \geq \sizeof{S}$, we give a set $X$ of $\sizeof{S}$ edges
  such that any clique in $G$ contains at most one edge in $X$. For each $i \in S$,
  the definition of $S$ implies that we may fix a vertex $i^* \in \n^+_i$ such that
  $\{i,i^*\} \not\subseteq \n^+_{i-1}$. By the definition of a $(\tp,\tm,w)$-PT graph, this
  yields $\{i,i^*\} \not\subseteq \n^+_r$ for all $r < i$. Let $X = \{ e_{ii^*} \st i \in S\}$.
  Clearly $\sizeof{X} = \sizeof{S}$.

  Now suppose that $i,j$ are distinct members of $S$, with $i < j$, and let $C$
  be a clique of $\G$ containing $e_{jj^*}$. The choice of $j^*$ implies that
  $\{j,j^*\} \not\subseteq \n^+_i$. Since $i < j < j^*$, we have $\{j, j^*\} \not\subseteq \n_i$,
  so $i \notin C$, and in particular $e_{ii^*} \notin E(C)$. Thus, every clique of $\G$
  contains at most one edge of $X$, so that $\iota(G) \geq \sizeof{S}$.
\end{proof}
\subsection{The Clustering Coefficient of PT Graphs}

The global clustering coefficient of a graph is defined based on counts of triplets of vertices~\cite{Holland,Watts}. A triplet consists of a vertex (center) and two distinct vertices that are adjacent to the center. A triplet is closed if the two vertices adjacent to the center are adjacent. A triangle in the graph includes three closed triplets, one centered on each of the vertices. 

Formally, the global clustering coefficient is defined as:
\begin{eqnarray}
   C = \frac{3 \times \mbox{\# of triangles}}{\mbox{\# of triplets}}
	= \frac{\mbox{\# of closed triplets}}{\mbox{\# of triplets}}.
	\label{Clus-Coef}
\end{eqnarray}	

To calculate the clustering coefficient of a PT graph, we assume that $G(V,E)$ is a connected PT graph with $|V|=n>1$ vertices, and assume that  the order of the vertices of $G$ has been established as $1, 2, \cdots, n$, such that $w(1) \leq w(2) \leq \cdots \leq
  w(n)$. In case that two vertices are assigned the same ranking within the order, we randomly break the tie. 
  
 Let $\{d_1, \cdots, d_n\}$ be a set in which $d_i$ is the degree of the $i$-th vertex in $G$, for $i=1, \cdots, n$. 
\begin{itemize}
\item It is straightforward to see that the number of triplets in the PT graph equals $\sum_{i=1}^n {d_i \choose 2}$. 
\item 
Recall the definition of $\n_i^+$ in Theorem~\ref{thm:nplus} and let $d_i^+=|\n_i^+|$. 
Then, the number of triangles in the PT graph equals $\sum_{i=1}^n {d_i^+ \choose 2}$. To see this, first consider a vertex $i$ and the set $\n_i^+$. Let $j_1,j_2 \in \n_i^+$,where $j_1<j_2$. Then, according to Corollary~\ref{RightNeighborsClique},  $j_1$ and $j_2$ are adjacent.
\end{itemize}
\begin{lemma}
Let $\G(\V,\e)$ be a connected PT graph with $n$ vertices. Assume that the order of the vertices of $G$ has been established as $1, 2, \cdots, n$ using (\ref{eq1}) and (\ref{hoof}). Let $d_i$ and $d_i^+$ be $|\n_i|$ and $|\n_i^+|$, respectively. Then,

\begin{equation}
C = \frac{3\times \sum_{i=1}^n {d_i^+ \choose 2}}{\sum_{i=1}^n {d_i \choose 2}}.
\end{equation}

\end{lemma} 
\textbf{Acknowledgment:} The authors gratefully acknowledge funding from the NIH BD2K Targeted Software Program, under the contract number U01 CA198943-02, and the NSF grants IOS1339388, CCF 1117980 and NSF Center for Science of Information STC Class 2009. 
The authors would also like to thank Hoang Dau, Pan Li and Hussein Tabatabei Yazdi at the University of Illinois for helpful discussions.

\end{document}